\documentclass[11pt]{article}

\usepackage{amstext,amsgen,latexsym,amsmath}
\usepackage{amssymb,amsfonts}
\usepackage{theorem}
\usepackage{pifont}
\usepackage[dvips]{graphics,epsfig}

 \newcommand{\bs}{\bigskip} 
  
 \newcommand{\hs}[1]{\hspace*{ #1 mm}}
 
 \def\bbox{\vrule height6pt width6pt depth1pt}

\theoremstyle{plain}
\theoremheaderfont{\bfseries}
 \newtheorem{theorem}{Theorem}[section] 
 \newtheorem{lemma}[theorem]{Lemma}

\newenvironment{proof}{\par \noindent
            {\bf Proof. \hs{2}}}{\hfill$\Box$ \vspace*{3mm}}
 
 \newenvironment{proofof}[1]{\vspace*{5mm} \par \noindent
         {\it Proof of #1.\hs{2}}}{ \vspace*{3mm}}

  \newcommand{\bra}[1]{\langle #1 |}
 \newcommand{\ket}[1]{| #1 \rangle} \newcommand{\bracket}[2]{\langle #1 | #2 \rangle}
 
\newcommand{\ignore}[1]{}
\begin{document}
%%%%%%%%%%%%%%%%%%
\pagestyle{plain}
\begin{center}
{\Large {\bf Unbounded-Error Classical and Quantum 

\medskip

Communication Complexity}}
\bs\\
\end{center}

%\begin{tabular}{c@{\hspace{30mm}}c}
\begin{center}
{\sc Kazuo Iwama}$^1$\footnote{Supported in part by Scientific Research Grant, Ministry of Japan, 16092101 and 19200001.} 
\hspace{5mm} {\sc Harumichi Nishimura}$^2$\footnote{Supported in part by Scientific Research Grant, 
Ministry of Japan, 19700011.}
\hspace{5mm} {\sc Rudy Raymond}$^3$ \hspace{5mm} {\sc Shigeru Yamashita}$^4$\footnote{Supported 
in part by Scientific Research Grant, Ministry of Japan, 16092218 and 19700010.}

$^1${School of Informatics, Kyoto University}, {\tt iwama@kuis.kyoto-u.ac.jp} 

$^2${School of Science, Osaka Prefecture University}, {\tt hnishimura@mi.s.osakafu-u.ac.jp}

$^3${Tokyo Research Laboratory, IBM Research}, {\tt raymond@jp.ibm.com}

$^4${Nara Institute of Science and Technology}, {\tt ger@is.naist.jp} 
\end{center}
\bs
\begin{abstract} 
Since the seminal work of Paturi and Simon \cite[FOCS'84 \& JCSS'86]{PS86}, the unbounded-error classical 
communication complexity of a Boolean function has been studied based on the arrangement 
of points and hyperplanes. Recently, \cite[ICALP'07]{INRY07} found that 
the unbounded-error {\em quantum} communication complexity in the {\em one-way communication} model 
can also be investigated using the arrangement, 
and showed that it is exactly (without a difference of even one qubit) half 
of the classical one-way communication complexity.  
In this paper, we extend the arrangement argument to the {\em two-way} and 
{\em simultaneous message passing} (SMP) models. 
As a result, we show similarly tight bounds of the unbounded-error two-way/one-way/SMP quantum/classical 
communication complexities for {\em any} partial/total Boolean function, implying that all of them are equivalent 
up to a multiplicative constant of four. 
Moreover, the arrangement argument is also used to show that the gap between {\em weakly} unbounded-error quantum and classical communication complexities is at most a factor of three. 
\end{abstract}

\section{Introduction}
As with many other probabilistic computation models, {\em communication complexity} 
(CC for short) has two contradistinctive settings: {\em Bounded-error} CC refers to the amount 
of communication (the number of bits exchanged) between Alice and Bob 
which is enough to compute a Boolean value $f(x,y)$, with high probability, 
from Alice's input $x$ and Bob's input $y$. On the other hand, {\em unbounded-error} CC 
refers to the lowest possible amount of communication which is needed 
to give ``a positive hint'' for the computation of $f(x,y)$, in other words, 
even one-bit less communication would be the same as completely no communication in the worst case. 
More formally, it is defined as the minimum amount of communication 
between Alice and Bob such that for all $x$ and $y$ Alice (or Bob) can 
output a correct value of $f(x,y)$ with probability $>1/2$. 

Unbounded-error CC was first studied by Paturi and Simon \cite{PS86}, 
who characterized its one-way version, $C^1(f)$, in terms of the minimum dimension $k_f$ 
of the {\em arrangement} that realizes the Boolean function $f$ 
(see Sec.\ \ref{sec-def} for the definition of arrangements). 
Namely they showed $\lceil\log k_f\rceil\leq C^1(f)\leq\lceil\log(k_f+2)\rceil$. 
It was also proven that the two-way (unbounded-error) CC, $C(f)$, 
does not differ from $C^1(f)$ more than one bit for any (partial or total) Boolean function $f$, 
which is a bit surprising since there are easily seen exponential differences between them 
in the bounded-error setting (see, say \cite{KN97}). 

Since then, arrangement has been a standard tool for studying unbounded-error CC. 
Alon, Frankl, and R\"{o}dl \cite{AFR85} showed by counting arguments that 
almost all Boolean functions have linear unbounded-error CCs. 
The first linear lower bound of an explicit function was found by Forster \cite{Fors02}, 
who gave the linear lower bound of the inner product function by showing the lower bound 
of its minimum dimension using operator norms. 
Extending Forster's arguments, there are several papers on the study 
of unbounded-error CC \cite{FKLMSS01,FS06} 
that also put emphasis on the margin of arrangements.  

Recently, \cite{INRY07} completely characterized the unbounded-error 
one-way (Alice to Bob) quantum CC, $Q^1(f)$, also in terms of $k_f$, i.e., 
$Q^1(f)=\lceil\log\sqrt{k_f+1}\rceil$. The main idea was to relate 
quantum states in Alice's side and POVMs in Bob's side to points and hyperplanes 
of a real space arrangement, respectively. 
Moreover, they also closed the small gap between the upper and lower bounds of $C^1(f)$ in \cite{PS86} 
by proving $C^1(f)=\lceil\log(k_f+1)\rceil$. As a result, 
they found that the unbounded-error one-way quantum CC 
of any Boolean function is always exactly one half of its classical counterpart.  
Unfortunately, however, their studies were limited within the one-way model: 
The proof technique mentioned above apparently depends on the one-way communication 
and there is no obvious way of its extension to the more general two-way communication model. 
Furthermore, it seems hard to change two-way {\em quantum} protocols to one-way quantum protocols 
efficiently, which was possible and was used as the basic approach in the classical case \cite{PS86}. 

{\bf Our Contribution.} We provide a new approach for constructing an arrangement from a given two-way quantum protocol with $n$ qubit communication. 
The basic idea is to use the simple fact, found by Yao \cite{Yao93} and Kremer \cite{Kre95}, 
that the final state of the whole system after the protocol is finished can be written as a superposition 
of at most $2^n$ different states. This allows us to imply a quite tight lower bound 
for the two-way quantum CC $Q(f)$, namely $Q(f)\geq \lceil{\log\sqrt{k_f+1/8}}-1/2\rceil$. 
Notice that this lower bound does not differ more than one qubit  
from the upper bound of one-way CC $Q^1(f)$ in \cite{INRY07}, which then means that all of $Q(f)$, $Q^1(f)$, $C(f)/2$ and $C^1(f)/2$ coincide within the difference of at most only one bit or one qubit. 

Arrangements are also useful to provide a couple of related results: 
First, we give almost tight characterizations of $Q^{||}(f)$ and $C^{||}(f)$, 
i.e., the unbounded-error quantum and classical CCs 
in the {\em simultaneous message passing} (SMP) model. We prove that 
$Q^{||}(f)$ and $C^{||}(f)$ are equal to twice as much as $Q^1(f)$ and $C^1(f)$ 
up to a few qubits and bits, respectively. Therefore we can see that in the unbounded-error setting 
all of the two-way/one-way/SMP quantum/classical CCs of any Boolean function 
are asymptotically equivalent up to a multiplicative constant of four. 
Note that, in the bounded-error classical case, the equality function gives 
an exponential gap between one-way and SMP CCs \cite{BK97,NS96}. 
In the bounded-error quantum case, it is also shown that an exponential gap 
between one-way and SMP CCs exists for some {\em relations} \cite{GKRW06}. 

Secondly, we give several relations among CCs 
in the {\em weakly unbounded-error} setting, which was introduced by Babai et al.\ \cite{BFS86}. 
The weakly unbounded-error (classical) CC of a protocol $P$, denoted by $C_w(P)$, 
is measured by the sum of the communication cost of $P$ and $\log 1/(p-1/2)$ if $P$'s success probability 
is $p$. The weakly unbounded-error CC of $f$, $C_w(f)$, 
is the minimum of $C_w(P)$ over all protocols $P$ that computes $f$. 
The quantum variant and one-way/SMP variants are defined similarly. 
Using two quantities of arrangement, margin and dimension, 
we show several upper bounds of weakly unbounded-error CCs, in particular, $C_w(f)\leq 3Q_w(f)+O(1)$. 
Previously, it is only known \cite{Kla01} that $C_w(f)=O(Q_w(f))$. 
The multiplicative factor three seems to be quite tight since at least a factor of two 
must be involved as a gap between quantum and classical communication costs as mentioned before. 

{\bf Related Work.} In the bounded-error setting, CCs of some Boolean functions 
have large gaps between quantum and classical cases: Exponential separations are known  
for all of two-way \cite{Raz99}, one-way \cite{GKKRW06} and SMP models \cite{BCWW01},  
where the first two cases are for partial Boolean functions, and the last case is 
for a total Boolean function. It remains to show (if any) exponential gaps 
for total Boolean functions in the cases of two-way and one-way models. 
In particular, the largest known gap between quantum and classical one-way CCs 
is only a factor of two.  
 
Other than the minimum dimension $k_f$ of arrangements, 
several different measures of Boolean functions also appeared in the literature. 
Paturi and Simon \cite{PS86} showed that $C^1(f)$ (and $C(f)$) 
is equal to the logarithm of the {\em sign-rank}, $srank(f)$, up to a few bits (also see \cite{BVW07}). 
Due to Klauck \cite{Kla01}, both $C_w(f)$ and $Q_w(f)$ are equivalent 
to the logarithm of the inverse of the {\em discrepancy} $disc(f)$ (see, say \cite{KN97}) 
within a constant multiplicative factor and a logarithmic additive factor. 
The recent result by Linial and Shraibman \cite{LS06,LS07} implies that 
the maximal margin of arrangements realizing $f$, $m(f)$, is equivalent to $disc(f)$ up to 
a multiplicative constant. Thus, combined with the results of the current paper, (i) $C(f)$, $Q(f)$, $\log k_f$ and 
$\log srank(f)$ are all within a factor of two, and (ii) $C_w(f)$, $Q_w(f)$, 
$\log disc^{-1}(f)$ and $\log m^{-1}(f)$ within a factor of some constant and 
a logarithmic additive term. However, due to the two independent results 
by Buhrman et al.\ \cite{BVW07} and Sherstov \cite{She07}, 
(i) is exponentially smaller than (ii) for some Boolean function $f$.

\section{Technical Components}\label{sec-def}
In this section, we present some basic tools for obtaining our results. Their proofs, as well as some of those in the following sections, are omitted due to space constraints. They are mainly the concept of arrangement and its sufficient conditions (Lemmas \ref{embed_states} and \ref{embed_meas}) for realizing a quantum protocol whose success probability can be calculated from arrangement parameters by Lemma 5 in \cite{INRY07}.
%We first give the definition of arrangements and show how to convert them into communication protocols. In particular, Lemmas \ref{embed_states} and \ref{embed_meas} show sufficient conditions for points and hyperplanes of an arrangement that represents a protocol whose success probability can be calculated by Lemma 5 of \cite{INRY07}.

{\bf Arrangements.} We denote a point in $\mathbb{R}^n$ by the corresponding 
$n$-dimensional real vector, and a hyperplane $\{(a_i)\in\mathbb{R}^n
\mid \sum_{i=1}^n a_ih_i=h_{n+1}\}$ by the $(n+1)$-dimensional real vector 
${\mathbf h}=(h_1,\ldots,h_n,h_{n+1})$, meaning that any point $(a_i)$ on the plane 
satisfies the equation $\sum_{i=1}^n a_ih_i =h_{n+1}$. 
A Boolean function $f$ on $X \times Y$ is {\it realizable by an arrangement} 
of a set of $|X|$ points ${\mathbf p}_x=(p_{1}^x,\ldots,p_{k}^x)$ 
and a set of $|Y|$ hyperplanes ${\mathbf h}_y =(h_{1}^y,\ldots,h_{k}^y,h_{k+1}^y)$ in $\mathbb{R}^k$ 
if for any $x \in X$ and $y \in Y$, $\mbox{sign}(\sum_{i=1}^k {p_{i}^x h_{i}^y} - h_{k+1}^y) = f(x,y)$. 
Here, $\mbox{sign}(a)=1$ if $a>0$ and $-1$ if $a<0$. 
The value $\left|\sum_{i=1}^k {p_{i}^x h_{i}^y} - h_{k+1}^y\right|$ denotes how far 
the point ${\mathbf p}_x$ lies from the plane ${\mathbf h}_y$, 
and the {\em margin} of an arrangement denotes the smallest of such values in the arrangement. 
The {\em magnitude} of the arrangement is defined as $\mbox{max}_{x,y}\left(\sqrt{\sum_{i=1}^k |p_i^x|^2},
\sqrt{\sum_{i=1}^k |h_i^y|^2},|h_{k+1}^y|\right)$. 
The value $k$ is called the {\em dimension} of the arrangement. 
Let $k_f$ denote the minimum dimension of all arrangements that realize $f$. 

{\bf Remark.} In the hereafter, our statements will use ``functions'' while their proofs, that obviously hold for partial, are showed only for total ones. Note also that the concept of arrangement in this paper is not {\it symmetric}. Here, Alice's input $x$ and Bob's input $y$ are associated with a point and a hyperplane, respectively. 
For this reason, the value of $k_f$ might be different from that of $k_{f^t}$, where $f^t(x,y):=f(y,x)$. However, it can be easily seen that $|k_f - k_{f^t}| \le 1$. The random access coding is one of examples such that $|k_f - k_{f^t}| = 1$ \cite{ANTV99,INRY07}.

The following lemma relates arrangements to classical CC, 
which was shown in \cite{PS86} and later in \cite{FKLMSS01} 
in more detail including the margin.

\begin{lemma}[From arrangements to classical CC]\label{fs01} 
Any $N$-dimensional arrangement realizing $f$ of magnitude at most $1$ 
with margin $\mu$ can be converted into a classical one-way protocol 
for $f$ using at most $\lceil\log(N+1)\rceil + 1$ bits with success probability 
at least $1/2 + \mu/(2\sqrt{N+1})$. 
\end{lemma}

{\bf Bloch Vector Representations of Quantum States.}
Let $N = 2^n$. Any $n$-qubit state can be represented by an $N\times N$ positive matrix $\pmb{\rho}$ 
(also often called $N$-level quantum state), satisfying $\mbox{Tr}(\pmb{\rho})=1$. 
Moreover, $\pmb{\rho}$ can be written as a linear combination of $N^2$ 
{\it generator matrices} $\pmb{I}_N,\pmb{\lambda}_1,\ldots,\pmb{\lambda}_{N^2-1}$, 
where $\pmb{I}_N$ is the identity matrix (the subscript $N$ is often omitted when it is clear from the context), 
and $\pmb{\lambda}_i$'s are $N \times N$ matrices which are generators of $SU(N)$ satisfying 
(i) $\pmb{\lambda}_i = \pmb{\lambda}_i^{\dagger}$ (i.e., $\pmb{\lambda}_i$'s are Hermitian), 
(ii) $\mbox{Tr}(\pmb{\lambda}_i) = 0$ and (iii) $\mbox{Tr}(\pmb{\lambda}_i \pmb{\lambda}_j) = 2\delta_{ij}$. 
Note that $\pmb{\lambda}_i$ can be any generator matrices satisfying the above conditions 
(and in fact $N$ can be any positive integer $\ge 2$), but practically for $n = 1$ 
one can choose $\pmb{\sigma}_1=
\left(\begin{array}{cc}
1 & 0\\
0 & -1
\end{array}
\right)$,
$\pmb{\sigma}_2=
\left(\begin{array}{cc}
0 & 1\\
1 & 0
\end{array}
\right)$, and
$\pmb{\sigma}_3=
\left(\begin{array}{cc}
0 & -\imath\\
\imath & 0
\end{array}
\right)$
of Pauli matrices as $\pmb{\lambda}_1,\pmb{\lambda}_2$, and $\pmb{\lambda}_3$, respectively. 
For larger $n$, one can choose the following tensor products of Pauli matrices
for $\pmb{\lambda}_1,\ldots,\pmb{\lambda}_{N^2-1}$:
$\pmb{\lambda}_1=\sqrt{\frac{2}{N}}\pmb{I}_2^{\otimes n-1}\otimes\pmb{\sigma}_1$,
$\pmb{\lambda}_2=\sqrt{\frac{2}{N}}\pmb{I}_2^{\otimes n-1}\otimes\pmb{\sigma}_2$,
$\pmb{\lambda}_3=\sqrt{\frac{2}{N}}\pmb{I}_2^{\otimes n-1}\otimes\pmb{\sigma}_3$,
$\pmb{\lambda}_4=\sqrt{\frac{2}{N}}\pmb{I}_2^{\otimes n-2}\otimes\pmb{\sigma}_1\otimes\pmb{I},
\ldots,\pmb{\lambda}_{N^2-2}=\sqrt{\frac{2}{N}}\pmb{\sigma}_3^{\otimes n-1}\otimes\pmb{\sigma}_2$,
and $\pmb{\lambda}_{N^2-1}=\sqrt{\frac{2}{N}}\pmb{\sigma}_3^{\otimes n}$. 
The following representation is known on $N$-level quantum states (see, e.g., \cite{KK04}).
\begin{lemma}\label{mstate_vec}
For any $N$-level quantum state $\pmb{\rho}$ and any $N\times N$ generator matrices $\pmb{\lambda}_i$'s,
there exists an $(N^2-1)$-dimensional vector ${\mathbf r}=(r_i)$ such that $\pmb{\rho}$  
can be written as
\[%\begin{equation}\label{mstate_vec_eq}
\pmb{\rho}= \frac{1}{N}\left(\pmb{I} +
\sqrt{\frac{N(N-1)}{2}}\sum_{i=1}^{N^2-1}r_i \pmb{\lambda}_i\right).
\]%\end{equation}
\end{lemma}
The vector ${\mathbf r}$ in this lemma is often called the {\em Bloch vector} of $\pmb{\rho}$.

Note that Lemma~\ref{mstate_vec} is a necessary condition for $\pmb{\rho}$
to be a quantum state. The following sufficient condition appeared in \cite{INRY07}, 
using the geometric fact of Bloch vectors in \cite{JS01,Koss03}.

\begin{lemma}\label{embed_states}
For any ${\mathbf r}=(r_1,r_2,\ldots,r_k)\in\mathbb{R}^k$ and any $N$ satisfying $N^2\geq k+1$,
%%let $\gamma= \frac{1}{|{\mathbf r}|(N-1)}$. Then,
\[
\pmb{\rho}({\mathbf r}) = \frac{1}{N}\left( \pmb{I} +\sqrt{\frac{N(N-1)}{2}} \sum_{i=1}^{k}
\left(\frac{r_i}{|{\mathbf r}|(N-1)}\right) \pmb{\lambda}_i \right)
\]
is an $N$-level quantum state.
(Intuitively, if a vector is shrunk enough to be inside the ball of radius $1/(N-1)$, 
its shrunk vector is always a quantum state.) Moreover, if $\pmb{\rho}({\mathbf r})$ 
is a quantum state, then $\pmb{\rho}(\gamma{\mathbf r})$ is also a quantum state 
for any $0 \le \gamma \leq 1$.
\end{lemma} 

{\bf Bloch Vector Representations of POVMs.} 
A POVM $M = \{\pmb{E},\pmb{I}-\pmb{E}\}$ is a set of operators, 
which represents a quantum measurement, such that $\pmb{E}$ and $\pmb{I}-\pmb{E}$ are positive matrices. 
It is known that any POVM $M$ on $N$-level quantum states 
can be written as a linear combination of $N \times N$ generator matrices 
$\pmb{\lambda}_i$'s. Namely, 
$$
\pmb{E} = e_{N^2}\pmb{I} + \sum_{i = 1}^{N^2-1} e_{i}\pmb{\lambda}_i, 
$$
where $\mathbf{e} = (e_1,e_2,\ldots,e_{N^2})$ is called the {\it Bloch vector representation} 
of POVM $M$. One sufficient condition for a vector to represent a POVM is given as follows. 
 
\begin{lemma}\label{embed_meas}
%% \begin{enumerate}
%%\item 
Let $\mathbf{e} = (e_1,e_2,\ldots,e_{N^2}) \in \mathbb{R}^{N^2}$ such that 
$$
\sum_{i=1}^{N^2-1} e_i^2 \le \frac{N}{2(N-1)}\mathrm{min}( e_{N^2}^2, (1- e_{N^2})^2 ). 
$$
If we take $\pmb{E} = e_{N^2}\pmb{I} +\sum_{i = 1}^{N^2-1} e_{i}\pmb{\lambda}_i$, 
then $\{\pmb{E},\pmb{I}-\pmb{E}\}$ is a POVM on $N$-level quantum states. 
%% \end{enumerate}
\end{lemma}

\section{Two-Way Communication Complexity}

The model is due to Yao \cite{Yao93}: The space of a quantum protocol consists 
of Alice and Bob's private parts and a communication channel. 
On her (his) turn, Alice (Bob) applies a unitary transformation 
on her (his) part and the communication channel, and Bob (Alice) receives 
quantum information from the content of the channel. 
Finally the output of the protocol is obtained by a measurement via Alice or Bob. 
Note that without loss of generality we can assume 
that no measurement is performed in the middle of the protocol. 
This is because it is well known that measurements 
can be postponed without increasing the communication cost \cite{NC00}. 
Also, it is often assumed, for technical reason, 
that the output is put on the communication channel. 
A protocol described under this output style (and Yao's formalism), which we call a {\em shared-output} protocol, 
means that the protocol's output can be known to {\em both} Alice and Bob.  
We define $Q(f)$ as the CC for {\em one of them} to know the output 
since we want to regard one-way protocols as a special case of two-way protocols. 
Thus our $Q(f)$ may be smaller than the corresponding CC under shared-output protocols, 
but we can easily see that the gap is at most one qubit. 

For the shared-output protocol, the following lemma, which was given by Yao \cite{Yao93} without proof and proved by Kremer \cite{Kre95}, is quite strong. 

\begin{lemma}[\cite{Yao93} and \cite{Kre95}]\label{yaokremer} 
The final state of a shared-output quantum protocol for a Boolean function $f$ 
on input $(x,y)$ using $n$ qubit of communication can be written as 
$$
\sum_{i\in\{0,1\}^n} \ket{A_i(x)}\ket{i_n}\ket{B_i(y)},
$$
where $\ket{A_i(x)}$ and $\ket{B_i(y)}$ are complex vectors of norm $\le 1$, 
and $i_n$ is the $n$th bit of the index $i$ and also the last bit of 
the communication channel (that is, the output bit). 
\end{lemma}

To see the intuitive meaning of this lemma might help understand the
proof of Lemma \ref{qc2ar} (our main lemma) more easily.  There are two
points: (i) The superposition consists of at most $2^n$ different
states, independent of the size of the whole space.  This allows us
to consider only $2^{2n}$ ($2^n \times 2^n$) different combinations of
vectors (and their inner product values) when calculating the trace of
underlying density matrices whose size may be much larger.  (ii) As
one can see, a product of state $A_i(x)$ and state $B_j(y)$ exists
only if $i=j$.  This correspondence is translated into the same
correspondence between the indices when calculating an inner product
of a point and a hyperplane of the converted arrangement. A similar correspondence was also used in \cite{BdW01} for lower bounds of quantum exact and bounded-error protocols, and in \cite{dW03} for tight lower bounds of quantum one-sided unbounded-error (which is referred as {\it nondeterministic}) protocols. 

Let $k_f^{*}=\mbox{min}(k_f,k_{f^t})$. Then here is our first main result. 

\begin{theorem}\label{thm1} 
For any Boolean function $f$, $\lceil\log\sqrt{k_f+1/8}-1/2\rceil 
\leq Q(f) \leq \lceil{\log{\sqrt{k_f^*+1}}}\rceil$. 
\end{theorem}

Theorem \ref{thm1} induces the equality of two-way and one-way 
quantum CCs of a Boolean function within one qubit since we can verify that the difference 
of the upper bound from the lower bound is at most one for any integer $k_f > 0$. 
Recall that the difference between the two-way and one-way CCs is also at most one 
in the classical case \cite{PS86}. 
For the proof of Theorem \ref{thm1}, it is enough to give the lower bound $Q(f)\geq 
\lceil\log\sqrt{k_f+1/8}-1/2\rceil$ since $Q(f)\leq \mathrm{min}(Q^1(f),Q^1(f^t))
=\lceil{\log{\sqrt{k_f^*+1}}}\rceil$. To do so, we relate quantum communication protocols to arrangements.

\begin{lemma}[From quantum CC to arrangements]\label{qc2ar} 
An $n$-qubit shared-output protocol that computes a Boolean function $f$ 
with success probability $1/2+\epsilon$ can be converted to 
a $(2^{2n-1}-2^{n-1})$-dimensional arrangement of magnitude at most $1$ 
that realizes $f$ with margin $\epsilon$. 
\end{lemma}

\begin{proof}
Suppose that $P$ is an $n$-qubit protocol for $f$. According to Lemma \ref{yaokremer}, 
we can write the final quantum state of $P$ on input $(x,y)$, $\pmb{\rho}_{xy}$, as follows. 
$$
\pmb{\rho}_{xy} = \sum_{i,j\in \{0,1\}^{n}} \ket{A_i(x)}\ket{i_n}\ket{B_i(y)}\bra{A_j(x)}\bra{j_n}\bra{B_j(y)}
= \pmb{\rho}^0_{xy} + \pmb{\rho}^1_{xy} + \widetilde{\pmb{\rho}}_{xy},
$$
where 
\begin{eqnarray*}
 \pmb{\rho}^0_{xy} &=& \sum_{i,j\in \{0,1\}^{n} ~\mbox{and}~i_{n} = j_{n} 
= 0} \ket{A_i(x)}\ket{0}\ket{B_i(y)}\bra{A_j(x)}\bra{0}\bra{B_j(y)},\\
 \pmb{\rho}^1_{xy} &=& \sum_{i,j\in \{0,1\}^{n} ~\mbox{and}~i_{n} = j_{n} 
= 1} \ket{A_i(x)}\ket{1}\ket{B_i(y)}\bra{A_j(x)}\bra{1}\bra{B_j(y)},
\end{eqnarray*}
and $\widetilde{\pmb{\rho}}_{xy} = \pmb{\rho}_{xy} - \pmb{\rho}^0_{xy}- \pmb{\rho}^1_{xy}$ 
such that $\mbox{Tr}(\pmb{\rho}_{xy})= \mbox{Tr}(\pmb{\rho}^0_{xy}) + \mbox{Tr}(\pmb{\rho}^1_{xy}) = 1$. 
Note that $\mbox{Tr}(\pmb{\rho}^0_{xy})$ (resp.\ $\mbox{Tr}(\pmb{\rho}^1_{xy})$) is the probability 
that the output of $P$ is $0$ (resp.\ $1$).  
By basic properties of the trace \cite{NC00}, 
$\mbox{Tr}(\pmb{\rho}^0_{xy})$ can be written as follows: $|m_A\rangle$ and $|m_B\rangle$ are the computational base 
of Alice's and Bob's spaces, respectively, and $b\in\{0,1\}$. Then,  
\begin{eqnarray*}
\mbox{Tr}(\pmb{\rho}^0_{xy}) 
&=& \sum_{m_A,b,m_B} \langle m_A|\langle b|\langle m_B|\pmb{\rho}^0_{xy}|m_A\rangle|b\rangle|m_B\rangle \\
&=& \sum_{m_A,m_B} \sum_{i,j\in\{0,1\}^{n-1}}
\langle m_A|\langle m_B|(\ket{A_{i0}(x)}\ket{B_{i0}(y)}\bra{A_{j0}(x)}\bra{B_{j0}(y)})|m_A\rangle|m_B\rangle \\ 
&=& \sum_{i,j\in\{0,1\}^{n-1}}\sum_{m_A,m_B} 
\bra{A_{j0}(x)}\bra{B_{j0}(y)}|m_A\rangle|m_B\rangle\langle m_A|\langle m_B|\ket{A_{i0}(x)}\ket{B_{i0}(y)} \\ 
&=& \sum_{i,j\in \{0,1\}^{n-1}} 
 \bracket{A_{j0}(x)}{A_{i0}(x)}\bracket{B_{j0}(y)}{B_{i0}(y)}, 
\end{eqnarray*}
where the last equation holds 
since $\sum_{m_A,m_B}|m_A\rangle|m_B\rangle\langle m_A|\langle m_B|=I$ (completeness relation).  
Now, let us define the following vectors $\mathbf{a}(x)\in\mathbb{C}^{2^{2n-2}}$ 
and $\mathbf{b}(y) \in \mathbb{C}^{2^{2n-2}+1}$. 
\begin{eqnarray*}
(\mathbf{a}(x))_k = (\mathbf{a}(x))_{ij} &=& \bracket{A_{j0}(x)}{A_{i0}(x)},\\
(\mathbf{b}(y))_k = (\mathbf{b}(y))_{ij} &=& \bracket{B_{j0}(y)}{B_{i0}(y)}\mbox{~for~} i,j \in \{0,1\}^{n-1},
\ \ (\mathbf{b}(y))_{2^{2n-2}+1} = 1/2,
\end{eqnarray*}
where the index $k\in[2^{2n-2}]$ naturally corresponds to the index $ij\in\{0,1\}^{2n-2}$. 
Since $P$ computes $f(x,y)$ with success probability $1/2+\epsilon$, 
$\mbox{Tr}(\pmb{\rho}^0_{xy}) \geq 1/2 + \epsilon$ if $f(x,y) = 0$ 
and $\leq 1/2 - \epsilon$ if $f(x,y) = 1$. Thus, the points $\mathbf{a}(x)$ 
and hyperplanes $\mathbf{b}(y)$ can be considered as an arrangement that ``realizes'' $f$ 
but they are in complex space. Fortunately, one can find an arrangement in $\mathbb{R}^{2^{2n-1}}$ 
that realizes $f$ from the above arrangement by noticing that $\mbox{Tr}(\pmb{\rho}^0_{xy})$ is always real. 
Namely, 
\begin{eqnarray}
\mbox{Tr}(\pmb{\rho}^0_{xy}) &=& \sum_{i,j\in \{0,1\}^{n-1}} 
\bracket{A_{j0}(x)}{A_{i0}(x)}\bracket{B_{j0}(y)}{B_{i0}(y)}\ =
 \sum_{k \in [2^{2n-2}]} (\mathbf{a}(x))_k (\mathbf{b}(y))_k \nonumber\\
&=& \sum_{k \in [2^{2n-2}]} \mbox{Re}\left((\mathbf{a}(x))_k 
				      (\mathbf{b}(y))_k\right) 
%= \sum_{k \in [2^{2n-1}]} (\mathbf{a}'(x))_k (\mathbf{b}'(y))_k 
\nonumber\\
&=& \sum_{k \in [2^{2n-2}]} \left(\mbox{Re}(\mathbf{a}(x))_k 
				      \mbox{Re}(\mathbf{b}(y))_k - 
				      \mbox{Im}(\mathbf{a}(x))_k \mbox{Im}(\mathbf{b}(y))_k\right)
\nonumber\\ %\label{eq01},
&=& \sum_{k \in [2^{2n-1}]} (\mathbf{a}'(x))_k (\mathbf{b}'(y))_k \label{eq01}, 
\end{eqnarray}
where 
\begin{eqnarray*}
& &(\mathbf{a}'(x))_{2k-1} = \mbox{Re}(\mathbf{a}(x))_k,\ \ (\mathbf{a}'(x))_{2k} 
= -\mbox{Im}(\mathbf{a}(x))_k, \\
& &(\mathbf{b}'(y))_{2k-1} = \mbox{Re}(\mathbf{b}(x))_k,
\ \ (\mathbf{b}'(y))_{2k}=\mbox{Im}(\mathbf{b}(x))_k,~\mbox{for}~k\in[2^{2n-2}], 
\end{eqnarray*}
and we set $(\mathbf{b}'(y))_{2^{2n-1}+1} = 1/2$. 
Now by Eq.(\ref{eq01}), the arrangement of points $\mathbf{a}'(x)$ and hyperplanes $\mathbf{b}'(y)$ 
realizes $f$ with margin $\epsilon$. Also, it is easy to see that its magnitude is at most $1$. 
Furthermore, since $\bracket{A_{i0}(x)}{A_{i0}(x)}$ and $\bracket{B_{j0}(y)}{B_{j0}(y)}$ 
are already real, the dimension of the above arrangement can be reduced from $2^{2n-1}$ to $2^{2n-1}-2^{n-1}$. 
\end{proof}

% Now we prove Theorem \ref{thm1}.

\begin{proofof}{Theorem \ref{thm1}}
Let $n=Q(f)$. As mentioned before Lemma \ref{yaokremer}, there exists an $(n+1)$-qubit shared-output protocol  
that computes $f$ with success probability larger than $1/2$. By Lemma \ref{qc2ar}, 
we can obtain a $(2^{2n+1}-2^{n})$-dimensional arrangement realizing $f$. 
Thus $k_f\leq 2(2^{n})^2-2^{n}$. By solving the quadratic inequality on $2^n$, 
$Q(f)=n\geq\lceil \log(\sqrt{8k_f+1}+1)\rceil -2$. The righthand side equals to 
$\lceil \log\sqrt{8k_f+1}\rceil -2=\lceil\log\sqrt{k_f+1/8}-1/2\rceil$ 
by a simple consideration on rounding reals, and hence we obtain the desired lower bound of $Q(f)$.  
On the contrary, it was proven that $Q^1(f)=\lceil \log\sqrt{k_f+1}\rceil$ \cite{INRY07}. 
Since $Q(f) \leq \mathrm{min}(Q^1(f),Q^1(f^t))$ (by our definition mentioned before Lemma \ref{yaokremer}), 
we obtain the desired upper bound. These complete the proof. 
\end{proofof}

\section{Simultaneous Message Passing Models}
The simultaneous message passing (SMP) model is the following three-party communication model: 
Alice and Bob have their inputs $x$ and $y$, respectively, but they have no interaction at all. 
The third party with no access to input, called the {\em referee}, must compute a Boolean function $f(x,y)$ 
with the help of two messages sent from Alice and Bob. For such a model, the corresponding CC are defined 
similarly to two-way or one-way CCs. 
   
We give quite tight characterizations of unbounded-error SMP CCs, $Q^{||}(f)$ and $C^{||}(f)$.  
First, we show the characterization of $Q^{||}(f)$ via $k_f$, which also implies that 
$Q^{||}(f)$ is the same as the sum of $Q^1(f)$ and $Q^{1}(f^t)$ up to two qubits. 

\begin{theorem}\label{qsmp}
For any Boolean function $f$, $Q^{1}(f)+Q^{1}(f^t)\leq Q^{||}(f)\leq Q^1(f)+Q^1(f^t)+2$.   
In particular, $$
\lceil\log\sqrt{k_{f}+1}\rceil+\lceil
\log\sqrt{k_{f^t}+1}\rceil\leq Q^{||}(f)\leq 2\lceil\log\sqrt{k_f^*+2}\rceil.
$$ 
\end{theorem}

\begin{proof} %of}{Theorem \ref{qsmp}}
For lower bound, $Q^{1}(f)+Q^{1}(f^t)\leq Q^{||}(f)$ is obtained 
by considering the relation between one-way communication models and SMP models: 
In the SMP model, Alice must send at least $Q^1(f)$ qubits to the referee.  
Otherwise, the number of qubits that she sends to the referee would be $m <Q^1(f)$, 
and then we can construct an $m$-qubit one-way protocol from Alice to Bob 
by regarding the referee and Bob as the same party, 
which contradicts the definition of $Q^1(f)$. Similarly Bob must send at least $Q^1(f^t)$ qubits. 
Since $Q^1(f)=\lceil\log\sqrt{k_{f}+1}\rceil$ for any $f$, 
we obtain $\lceil\log\sqrt{k_{f}+1}\rceil+\lceil\log\sqrt{k_{f^t}+1}\rceil \leq Q^{||}(f)$, 
and $2\lceil\log\sqrt{k_{f}^*+2}\rceil \leq Q^{1}(f)+Q^{1}(f^t)+2$. 

What remains to do is to show the upper bound $Q^{||}(f)\leq 2\lceil\log\sqrt{k_{f}^*+2}\rceil$. 
For this purpose, we can use {\em quantum fingerprinting} introduced in \cite{BCWW01}. 
That is, Alice's input $x$ and Bob's $y$ are encoded into two quantum states 
$\pmb{\rho}_x$ and $\pmb{\rho}_y$, respectively, and the referee uses the controlled SWAP (C-SWAP) test. 
The difference from the standard quantum fingerprinting such as \cite{BCWW01,Yao03,GKW06}  
is that we use mixed states for encoding. 
%, which seems to be usable only in the unbounded-error setting. 
(The C-SWAP test for mixed states are also used in \cite{KMY03} for quantum Merlin-Arthur games.)

We assume $k_f\leq k_{f^t}$ and show $Q^{||}(f)\leq 2\lceil\log\sqrt{k_{f}+2}\rceil$. 
(The case of $k_f > k_{f^t}$ is similarly shown.) Let $d=k_f$. 
Then there is an arrangement of points $\pmb{p}_x=(p_i^x)\in\mathbb{R}^d$ 
and hyperplanes $\pmb{h}_y=(h_i^y)\in\mathbb{R}^{d+1}$ that realizes $f$. 
Let $n=\lceil\log\sqrt{d+2}\rceil$ and $N=2^n$. Also, for each $x$, define 
$\pmb{q}_x=(q_i^x)\in\mathbb{R}^{d+1}$ as $q_1^x=p_1^x,\ldots,q_d^x=p_d^x$, $q_{d+1}^x=-1$. 
By Lemma \ref{embed_states}, for each $\pmb{q}_x$ and $\pmb{h}_y$ we can obtain 
$n$-qubit states $\pmb{\rho}(\pmb{q}_x)=\frac{1}{N}\left(\pmb{I}+\sqrt{\frac{N(N-1)}{2} }
\sum_{i=1}^{d+1} \left( \frac{q_i^x}{|\pmb{q}_x|(N-1)} \right)\pmb{\lambda}_i \right)$ 
and $\pmb{\rho}(\pmb{h}_y)= \frac{1}{N}\left(\pmb{I}+\sqrt{\frac{N(N-1)}{2} }
\sum_{i=1}^{d+1} \left( \frac{h_i^y}{|\pmb{h}_y|(N-1)} \right)\pmb{\lambda}_i \right)$. 
Then, we consider the following SMP quantum protocol: 
(1) Alice and Bob send the referee $\pmb{\rho}(\pmb{q}_x)$ and $\pmb{\rho}(\pmb{h}_y)$, respectively. 
(2) The referee outputs the bit obtained by the C-SWAP test on the pair 
of the quantum states $(\pmb{\rho}(\pmb{q}_x),\pmb{\rho}(\pmb{h}_y))$ with probability 
$\alpha=\frac{1}{2}\left(\frac{1}{2}+\frac{1}{2N}\right)^{-1}$, and otherwise 
outputs $1$ with probability $1-\alpha$.  
Note that the C-SWAP test produces output $0$ with probability $\frac{1}{2}+\frac{1}{2}
\mathrm{Tr}(\pmb{\rho}(\pmb{q}_x)\pmb{\rho}(\pmb{h}_y))$ \cite{BCWW01,KMY03}. 
Thus, the referee outputs $0$ with probability
\begin{eqnarray*}
\alpha\left(
\frac{1}{2}+\frac{1}{2}\mathrm{Tr}(\pmb{\rho}(\pmb{q}_x)\pmb{\rho}(\pmb{h}_y))\right)
&=& 
\frac{1}{2}\left(\frac{1}{2}+\frac{1}{2N}\right)^{-1}
\left( \frac{1}{2}+\frac{1}{2N}+\frac{N-1}{2N}\sum_{i=1}^{d+1}\frac{q_i^x h_i^y}{|\pmb{q}_x||\pmb{h}_y|(N-1)^2}
\right) \\
&=& \frac{1}{2}+\frac{1}{4N|\pmb{q}_x||\pmb{h}_y|(N-1)}
\left(\frac{1}{2}+\frac{1}{2N}\right)^{-1}\left(\sum_{i=1}^{d} p_i^x h_i^y - h_{d+1}^y\right)\\
&=& \left\{ 
\begin{array}{l}
>1/2\ \mbox{if}\ f(x,y)=0 \\
<1/2\ \mbox{if}\ f(x,y)=1.
\end{array}
\right.
\end{eqnarray*}

Hence $Q^{||}(f)\leq 2n=2\lceil\log\sqrt{d+2}\rceil$. 
\end{proof}

Moreover, we can also show a similar result in the classical setting.  

\begin{theorem}\label{csmp}
For any Boolean function $f$, $C^{1}(f)+C^{1}(f^t)\leq C^{||}(f)\leq C^{1}(f)+C^{1}(f^t)+1$.   
In particular, $$
\lceil\log(k_{f}+1)\rceil+\lceil
\log(    k_{f^t}+1)\rceil \leq C^{||}(f)\leq \lceil\log(
k_f^*+1)\rceil+\lceil\log(k_f^*+2)\rceil.
$$ 
\end{theorem}

\section{Weakly Unbounded-Error Communication Complexity}
Finally we give several relations among the weakly unbounded-error CCs. 
For this purpose, we need Lemmas \ref{fs01}, \ref{qc2ar} and \ref{ar2qc} 
to consider the bias of the success probability explicitly when converting protocols to arrangement, and vice versa. 
%%By these lemmas, we can show:
%From these we have: 

\begin{theorem}\label{thm3} The following relations hold for any Boolean function $f$:  
(1) $C_w(f)$ $\leq C_w^1(f)\leq 3Q_w(f)+O(1)$. (2) $Q_w^1(f)\leq 2Q_w(f)+O(1)$. 
\end{theorem}

\begin{proof}
1) By the definition of $Q_w(f)$, there is a quantum protocol $P$ 
such that $Q_w(f) = C_P + \lceil \log 1/\epsilon_P \rceil$ where $C_P$ 
and $1/2+\epsilon_P$ are the communication cost and the success probability of $P$, 
respectively. By Lemma \ref{qc2ar}, we can obtain a $(2^{2C_P-1}-2^{C_P-1})$-dimensional arrangement 
of magnitude at most 1 with margin $\epsilon_P$ from $P$. 
By Lemma \ref{fs01}, we have a $2C_P$-bit one-way protocol that 
computes $f$ with probability $\geq 1/2+\epsilon_P/(2\sqrt{2^{2C_P-1}})$. 
This implies that $C_w^1(f)\leq 2C_P+\lceil \log (2\sqrt{2^{2C_P-1}}/\epsilon_P)\rceil$, 
which is at most $3C_P+\lceil\log 1/\epsilon_P\rceil+O(1)\leq 3Q_w(f)+O(1)$. 

2) The proof idea is similar to 1). The difference from 1) is to construct 
a desired protocol from the arrangement. To this end, we use 
the following lemma, whose proof is omitted, that convert arrangements to one-way quantum CC. The proof follows from carefully transforming points and hyperplanes, with appropriate shrinking and shifting factors, to quantum states (by Lemma \ref{embed_states}) and measurements (by Lemma \ref{embed_meas}), respectively. The success probabilities of resulting protocols then follows from Lemma~5 of \cite{INRY07}.

\begin{lemma}[From arrangements to quantum CC]\label{ar2qc} 
Each $d$-dimensional arrangement of magnitude at most 1 realizing $f$ with margin $\mu$ 
can be converted into an $n=\lceil\log\sqrt{d+1}\rceil$ qubit one-way protocol 
that computes $f$ with success probability at least $1/2+\alpha\mu$ where $\alpha=\frac{\sqrt{2}-1}{2^{n+1/2}}$.  
\end{lemma}

Now we give the proof of Theorem \ref{thm3} (2). Take a quantum protocol $P$ 
such that $Q_w(f) = C_P +\lceil \log 1/\epsilon_P \rceil$ where $C_P$ and 
$1/2+\epsilon_P$ are the communication cost and the success probability of $P$, respectively. 
By Lemma \ref{qc2ar}, we can obtain a $(2^{2C_P-1}-2^{C_P-1})$-dimensional arrangement 
of magnitude at most 1 with margin $\epsilon_P$ from $P$. By Lemma \ref{ar2qc}, 
we have a one-way quantum protocol for $f$ using at most $C_P$ qubits 
such that its success probability is $1/2+\Omega(\epsilon_P/2^{C_P})$. 
This implies that $Q^1_w(f) \le 2C_P +\lceil \log{ 1/\epsilon_P} 
\rceil+O(1)\leq 2Q_w(f)+O(1)$. 
\end{proof}

%How about relations between $Q_w(f)$ and the weakly unbounded-error SMP CC 
%$Q^{||}_w(f)$ and $C^{||}_w(f)$? 
Similar to the proof of Theorem \ref{thm3}, using the proofs of Theorems \ref{qsmp} and \ref{csmp} we can also show: $Q_w^{||}(f)\leq 4Q_w(f)+O(1)$ and $C_w^{||}(f)\leq 9Q_w(f)+O(1)$ . 
%instead of Lemma \ref{fs01} or Lemma \ref{ar2qc}).

\section*{Acknowledgements}
We would like to acknowledge Francois Le Gall for insightful discussion on the two-way model, and Prof. Hiroshi Imai of University of Tokyo and ERATO-SORST project for partial support that enabled us to have a useful face-to-face discussion with quantum computation and information researchers in Japan while writing the paper.

%\vspace{-4mm}
%\begin{thebibliography}{99}

\end{document}